\documentclass{article}

\usepackage{fullpage}
\usepackage{tikz}
\usetikzlibrary{snakes}
\usepackage{xspace}
\usepackage{comment}

\usepackage{makeidx}  % allows for indexgeneration

\usepackage{amsmath,amssymb,amsthm}

\usepackage{todonotes}

\newtheorem{theorem}{Theorem}[section]
\newtheorem{lemma}[theorem]{Lemma}
\newtheorem{definition}[theorem]{Definition}

\newtheorem{observation}[theorem]{Observation}

\begin{document}

\newcommand{\Oo}{\mathcal{O}}
\newcommand{\Ohstar}{\Oo^\star}
\newcommand{\eps}{\varepsilon}
\newcommand{\vissets}{{\rm{deg2sets}}}
\newcommand{\pathsets}{{\rm{path}}}

\newcommand{\defproblem}[4]{
%  \vspace{1mm}
%  \hline
  \vspace{2mm}
\noindent\fbox{
  \begin{minipage}{0.96\textwidth}
  \begin{tabular*}{0.96\textwidth}{@{\extracolsep{\fill}}lr} #1 & {\bf{Parameter:}} #3 \\ \end{tabular*}
  {\bf{Input:}} #2  \\
  {\bf{Question:}} #4
  \end{minipage}
  }
%  \vspace{1mm}
%  \hline
  \vspace{2mm}
}
\newcommand{\defnoparamproblem}[3]{
%  \vspace{1mm}
%  \hline
  \vspace{2mm}
\noindent\fbox{
  \begin{minipage}{0.96\textwidth}
  #1
  {\bf{Input:}} #2  \\
  {\bf{Question:}} #3
  \end{minipage}
  }
%  \vspace{1mm}
%  \hline
  \vspace{2mm}
}

\newcommand{\Vop}[2]{\ensuremath{V_{\deg#1#2}}}
\newcommand{\Veq}[1]{\Vop{=}{#1}}
\newcommand{\Vgeq}[1]{\Vop{\geq}{#1}}
\newcommand{\Vleq}[1]{\Vop{\leq}{#1}}
\newcommand{\Vgt}[1]{\Vop{>}{#1}}

\newcommand{\bipconst}{3.55}

\title{Faster exponential-time algorithms in graphs of bounded average degree\thanks{Partially supported by NCN grant N206567140 and Foundation for Polish Science.}}

\author{
  Marek Cygan\thanks{Institute of Informatics, University of Warsaw, Poland, \texttt{cygan@mimuw.edu.pl}}
  \and
  Marcin Pilipczuk\thanks{Institute of Informatics, University of Warsaw, Poland, \texttt{malcin@mimuw.edu.pl}}
}
%\authorrunning{}
%\titlerunning{}

\maketitle

\begin{abstract}
We first show that the Traveling Salesman Problem in an $n$-vertex graph with average degree bounded
by $d$ can be solved in $\Ohstar(2^{(1-\eps_d)n})$ time\footnote{The $\Ohstar$-notation suppresses factors
polynomial in the input size.} and exponential space
for a constant $\eps_d$ depending only on $d$.
Thus, we generalize the recent results of Bj\"{o}rklund et al. [TALG 2012] on graphs of bounded degree. 

Then, we move to the problem of counting perfect matchings in a graph.
We first present a simple algorithm for counting perfect matchings in an $n$-vertex graph
in $\Ohstar(2^{n/2})$ time and polynomial space; our algorithm matches the complexity bounds
of the algorithm of Bj\"orklund [SODA 2012], but relies on inclusion-exclusion principle instead
of algebraic transformations. Building upon this result, we show that the number of perfect matchings
in an $n$-vertex graph with average degree bounded by $d$ can be computed
in $\Ohstar(2^{(1-\eps_{2d})n/2})$ time and exponential space, where $\eps_{2d}$ is the constant
obtained by us for the Traveling Salesman Problem in graphs of average degree at most $2d$.

Moreover we obtain a simple algorithm that counts the number of
perfect matchings in an $n$-vertex {\em{bipartite}} graph of average degree at most $d$
in $\Ohstar(2^{(1-1/(\bipconst d))n/2})$ time, improving and simplifying
the recent result of Izumi and Wadayama [FOCS 2012].
\end{abstract}

\section{Introduction}

Improving upon the 50-years old $\Ohstar(2^n)$-time dynamic programming algorithms
for the Traveling Salesman Problem by Bellman \cite{bellman} and Held and Karp \cite{held-karp}
is a major open problem in the field of exponential-time algorithms \cite{woeginger01}.
A similar situation appears when we want to count perfect matchings in the graph:
a half-century old $\Ohstar(2^{n/2})$-time algorithm of Ryser for bipartite graphs \cite{ryser}
has only recently been transferred to arbitrary graphs \cite{bjorklund-matchings},
and breaking these time complexity barriers seems like a very challenging task.

From a broader perspective, improving upon a trivial brute-force or a simple dynamic programming algorithm
is one of the main goals the field of exponential-time algorithms.
Although the last few years brought a number of positive results in that direction, most notably
the $\Ohstar(1.66^n)$ randomized algorithm for finding a Hamiltonian cycle in an undirected graph
\cite{hamilton}, it is conjectured (the so-called Strong Exponential Time Hypothesis \cite{seth}) that the
problem of satisfying a general CNF-SAT formulae does not admit any exponentially better algorithm than the trivial brute-force one.
A number of lower bounds were proven using this assumption \cite{ccc,lms:tw,patrascu-williams}.

In 2008 Bj\"{o}rklund et al. \cite{tsp-deg} observed that the classical dynamic programming
algorithm for TSP can be trimmed to running time $\Ohstar(2^{(1-\eps_\Delta)n})$ in
graphs of maximum degree $\Delta$. The cost of this improvement is the use of exponential
space, as we can no longer easily translate the dynamic programming algorithm into an inclusion-exclusion
formula. The ideas from \cite{tsp-deg} were also applied to the Fast Subset Convolution algorithm,
yielding a similar improvements for the problem of computing the chromatic number in graphs
of bounded degree \cite{trimmed-fsc}.
In this work, we investigate the class of graphs of bounded {\em{average}}
degree, a significantly broader graph class than this of bounded maximum degree.

In the first part of our paper we generalize the results of \cite{tsp-deg}.
\begin{theorem}\label{thm:tsp}
For every $d \geq 1$ there exists a constant $\eps_d > 0$
such that, given an $n$-vertex graph $G$ of average degree bounded by $d$,
in $\Ohstar(2^{(1-\eps_d)n})$ time and exponential space
one can find in $G$ a smallest weight Hamiltonian cycle.
\end{theorem}
We note that in Theorem \ref{thm:tsp} the constant $\eps_d$ depends on $d$ in doubly-exponential
manner, which is worse than the single-exponential behaviour of \cite{tsp-deg} in graphs
of bounded degree.

The proof of Theorem \ref{thm:tsp} follows the same general approach as the results of \cite{tsp-deg} ---
we want to limit the number of states of the classical dynamic programming algorithm for TSP
--- but, in order to deal with graphs of bounded {\em{average}} degree, we need to introduce new concepts
and tools.
Recall that, by standard averaging argument, if the average degree of an $n$-vertex graph $G$ is bounded by $d$,
for any $D \geq d$ there are at most $dn/D$ vertices of degree at least $D$. However,
it turns out that this bound cannot be tight for a large number of values of $D$ at the same time.
This simple observation lies at the heart of the proof of Theorem \ref{thm:tsp}, as we may afford
more expensive branching on vertices of degree more than $D$ provided that there are significantly less than
$dn/D$ of them.

In the second part, we move to the problem of counting perfect matchings in an $n$-vertex graph.
We start with an observation that this problem can be reduced to a problem of counting
some special types of cycle covers, which, in turn, can be done in $\Ohstar(2^{n/2})$-time
and polynomial space using the inclusion-exclusion principle (see Section~\ref{sec:in-ex-matchings}).
Note that an algorithm matching this bound in general graphs 
has been discovered only last year \cite{bjorklund-matchings},
in contrast to the 50-years old algorithm of Ryser \cite{ryser} for bipartite graphs.
Thus, we obtain a new proof of the main result of \cite{bjorklund-matchings},
using the inclusion-exclusion principle instead of advanced algebraic transformations.

Once we develop our inclusion-exclusion-based algorithm for counting perfect matchings,
we may turn it into a dynamic programming algorithm
and apply the ideas of Theorem \ref{thm:tsp}, obtaining the following.
\begin{theorem}\label{thm:matchings}
Given an $n$-vertex graph $G$ of average degree bounded by $d$,
in $\Ohstar(2^{(1-\eps_{2d})n/2})$ time and exponential space
one can count the number of perfect matchings in $G$
where $\eps_{2d}$ is the constant given by Theorem \ref{thm:tsp}
for graphs of average degree at most $2d$.
\end{theorem}
To the best of our knowledge, this is the first result that
breaks the $2^{n/2}$-barrier for counting perfect matchings in not necessarily bipartite
graphs of bounded (average) degree.

When bipartite graphs are concerned, the classical algorithm of Ryser \cite{ryser}
has been improved for graphs of bounded average degree first by Servedio and Wan \cite{servedio-wan}
and, very recently, by Izumi and Wadayama \cite{japonce}. 
Our last result is the following theorem.
\begin{theorem}\label{thm:bip-matchings}
Given an $n$-vertex bipartite graph $G$ of average degree bounded by $d$,
in $\Ohstar(2^{(1-1/(\bipconst d))n/2})$ time and exponential space
one can count the number of perfect matchings in $G$.
\end{theorem}
Hence, we improve the running time of \cite{japonce,servedio-wan} in terms of the dependency on $d$.
We would like to emphasise that our proof of Theorem \ref{thm:bip-matchings} is elementary
and does not need the advanced techniques of coding theory used in \cite{japonce}.

\paragraph{Organization of the paper}
Section~\ref{sec:prelim} contains preliminaries.
Next, in Section~\ref{sec:properties} we prove the main technical tool,
that is Lemma~\ref{lem:vissets-bound},
used in the proofs of Theorem~\ref{thm:tsp} and Theorem~\ref{thm:matchings}.
In Section~\ref{sec:tsp} we prove Theorem~\ref{thm:tsp},
while in Section~\ref{sec:in-ex-matchings} we first show
an inclusion-exclusion based algorithm for counting
perfect matchings, which is later modified in Section~\ref{sec:dp-matchings}
to fit the bounded average degree framework and prove Theorem~\ref{thm:matchings}.
Finally, Section~\ref{sec:bip-matchings} contains a simple
dynamic programming algorithm, proving Theorem~\ref{thm:bip-matchings}.

We would like to note that both Section~\ref{sec:in-ex-matchings} and Section~\ref{sec:bip-matchings}
are self-contained and do not rely on other sections (in particular
do not depend on Lemma~\ref{lem:vissets-bound}).

\section{Preliminaries}
\label{sec:prelim}

We use standard (multi)graph notation. For a graph $G=(V,E)$ and a vertex $v \in V$
the neighbourhood of $v$ is defined as $N_G(v) = \{u: uv \in E\} \setminus \{v\}$ and
the closed neighbourhood of $v$ as $N_G[v] = N_G(v) \cup \{v\}$.
The degree of $v \in V$ is denoted $\deg_G(v)$ and equals
the number of end-points of edges incident to $v$.
In particular a self-loop contributes $2$ to the degree of a vertex.
We omit the subscript if the graph $G$ is clear from the context.
The average degree of an $n$-vertex graph $G=(V,E)$ is defined as $\frac{1}{n} \sum_{v \in V} \deg(v) = 2|E|/n$.
A cycle cover in a multigraph $G=(V,E)$ 
is a subset of edges $C \subseteq E$,
where each vertex is of degree exactly two if $G$ is undirected
or each vertex has exactly one outgoing and one ingoing arc, if $G$ is directed.
Note that this definition allows a cycle cover to contain cycles of length $1$, i.e. self-loops, as well as taking two different parallel edges as length $2$ cycle
(but does not allow using twice the same edge).

For a graph $G=(V,E)$ by $\Veq{c}, \Vgt{c}, \Vgeq{c}$ let us denote
the subsets of vertices of degree equal to $c$, greater than $c$ and at least $c$ respectively.

We also need the following well-known bounds.
\begin{lemma}\label{lem:coeff-bound}
For any $n, k \geq 1$ it holds that
$$\binom{n}{k} \leq \left( \frac{en}{k} \right)^k.$$
\end{lemma}
\begin{lemma}\label{lem:Hn-bound}
For any $n \geq 1$, it holds that $H_{n-1} \geq \ln n$, where $H_n = \sum_{i=1}^n \frac{1}{i}$.
\end{lemma}
\begin{proof}
It is well-known that $\lim_{n \to \infty} H_n - \ln n = \gamma$
where $\gamma > 0.577$ is the Euler-Mascheroni constant and the sequence $H_n-\ln n$ is decreasing.
Therefore $H_{n-1} = H_n - \frac{1}{n} \geq \ln n + \gamma - \frac{1}{n}$, hence the lemma is proven for $n \geq 2$
as $\gamma > \frac{1}{2}$. For $n=1$, note that $H_{n-1} = \ln n = 0$.
\end{proof}

\section{Properties of bounded average degree graphs}
\label{sec:properties}

This section contains technical results concerning bounded
average degree graphs.
In particular we prove Lemma~\ref{lem:vissets-bound},
which is needed to get the claimed running times 
in Theorems~\ref{thm:tsp} and~\ref{thm:matchings}.
However, as the proofs of this section are not needed 
to understand the algorithms in further sections
the reader may decide to see only Definition~\ref{def:vissets}
and the statement of Lemma~\ref{lem:vissets-bound}.

\begin{lemma}
\label{lem:disjoint}
Given an $n$-vertex graph $G=(V,E)$ of average degree at most $d$ and maximum degree at most $D$
one can in polynomial time find a set $A$ containing $\lceil \frac{n}{2+4dD} \rceil$ vertices
of degree at most $2d$, where for each $x,y \in A$, $x \neq y$ we have $N_G[x] \cap N_G[y] = \emptyset$.
\end{lemma}

\begin{proof}
Note that $|\Vleq{2d}| \geq n/2$. We apply the following procedure.
Initially we set $A:=\emptyset$ and all the vertices are unmarked.
Next, as long as there exists an unmarked vertex $x$ in $\Vleq{2d}$,
we add $x$ to $A$ and mark all the vertices $N_G[N_G[x]]$.
Since the set $N_G[N_G[x]]$ contains at most $1+2d+2d(D-1) = 1+2dD$ vertices,
at the end of the process we have $|A| \geq \frac{n}{2+4dD}$.
Clearly this routine can be implemented in polynomial time.
\end{proof}

\begin{lemma}
\label{lem:gap}
For any $\alpha \geq 0$
and an $n$-vertex graph $G=(V,E)$ of average degree at most $d$
there exists $D \leq e^\alpha $ such that $|\Vgt{D}| \leq \frac{nd}{\alpha D}$.
\end{lemma}

\begin{proof}
By standard counting arguments we have
$$\sum_{i=0}^\infty |\Vgt{i}| = \sum_{i=0}^\infty i|\Veq{i}| \leq nd.$$
For the sake of contradiction assume that
$|\Vgt{i}| > \frac{nd}{\alpha i}$, 
for each $i \leq e^\alpha$. Then
$$\sum_{i=0}^\infty |\Vgt{i}| \geq \sum_{i=1}^{\lfloor e^\alpha \rfloor} |\Vgt{i}| > \frac{nd}{\alpha} \sum_{i=1}^{\lfloor e^\alpha \rfloor} 1/i = \frac{nd}{\alpha} H_{\lfloor e^\alpha \rfloor} \geq nd,$$
where the last inequality follows from Lemma \ref{lem:Hn-bound}.
\end{proof}

In the following definition we capture the superset of the
sets used in the dynamic programming algorithms of Theorems~\ref{thm:tsp}
and~\ref{thm:matchings}.

\begin{definition}
\label{def:vissets}
For an undirected graph $G=(V,E)$ and two vertices $s,t \in V$
by $\vissets(G,s,t)$ we define the set of all subsets $X \subseteq V \setminus \{s,t\}$,
for which there exists a set of edges $F \subseteq E$ such that:
\begin{itemize}
  \item $\deg_F(v)=0$ for each $v \in V \setminus (X \cup \{s,t\})$,
  \item $\deg_F(v)=2$ for each $v \in X$,
  \item $\deg_F(v) \le 1$ for $v \in \{s,t\}$.
\end{itemize}
\end{definition}

\begin{lemma}\label{lem:vissets-bound}
For every $d \ge 1$ there exists a constant $\eps_d>0$,
such that for an $n$-vertex graph $G=(V,E)$ 
of average degree at most $d$ for any $s,t \in V$ the cardinality of $\vissets(G,s,t)$
is at most $\Ohstar(2^{(1-\eps_d)n})$.
\end{lemma}

\begin{proof}
Use Lemma~\ref{lem:gap} with $\alpha=e^{cd}$ for some sufficiently large universal constant $c$ (it suffices to take $c=20$).
Hence we can find an integer $D \leq e^\alpha= e^{e^{cd}}$ such
that there are at most $\frac{nd}{\alpha D}$ vertices
of degree greater than $D$ in $G$.

Let $D' = \max(2d,D)$ and $H=G[\Vleq{D'}]$. 
Moreover let $Y = \Vgt{D'}$ and recall $|Y| \leq \frac{nd}{\alpha D}$, as $D' \geq D$ and $Y \subseteq \Vgt{D}$.
Note that $H$ contains at least $n/2$ vertices and has average degree upper bounded by $d$.
By Lemma~\ref{lem:disjoint} there exists a set $A \subseteq V(H)$
of $\lceil n/(4+8dD') \rceil$ vertices having disjoint closed neighbourhoods in $H$.
Note that, since $d \geq 1$ and $D' \geq 2d$:
\begin{equation}\label{eq:Adown}
|A| = \left\lceil \frac{n}{4+8dD'} \right\rceil \geq \frac{n}{4+8dD'} \geq \frac{n}{2dD'+8dD'} = \frac{n}{10dD'}.
\end{equation}
If $n \leq \frac{8edD'}{4-e}$, $n = \Oo(1)$ and the claim is trivial. Otherwise:
\begin{equation}\label{eq:Aup}
|A| = \left\lceil \frac{n}{4+8dD'} \right\rceil < \frac{n}{8dD'} + 1 < \frac{n}{2edD'}.
\end{equation}
Moreover, as $d \geq 1$ and $D' = \max(2d,D) \leq 2dD$, for sufficiently large $c$ we have:
\begin{equation}\label{eq:YA}
|Y| \leq \frac{nd}{\alpha D} \leq \frac{n}{20dD'} \cdot \frac{40d^3}{e^{cd}} < \frac{n}{20dD'} \leq \frac{|A|}{2}.
\end{equation}

Consider an arbitrary set $X \in \vissets(G,s,t)$, and a corresponding set $F \subseteq E$
from Definition~\ref{def:vissets}.
Define $Z_{X}$ as the set of
vertices $x \in X \cap V(H)$ such that $N_{H}(x) \cap X = \emptyset$.
Note that $F$ is a set of paths and cycles, where each vertex of $Z_X$ is of degree two,
hence $F$ contains at least $2|Z_X|$ edges between $Z_X$ and $Y$,
as any path/cycle of $F$ visiting a vertex of $Z_X$ has to enter from $Y$ and leave to $Y$.
Hence by the upper bound of $2$ on the degrees in $F$ we have $|Z_X| \le |Y|$.

For each $x \in A \setminus (Z_{X} \cup \{s,t\})$ we have
that $N_H[x] \cap X \neq \{x\}$ and $|N_H[x]| \leq 2d+1$.
By definition, if $x \in A \cap Z_{X}$, we have $N_H[x] \cap X = \{x\}$.
Therefore, for fixed $v$ and $A \cap Z_{X}$ there are at most
$$2^n \left( \frac{2^{2d+1}-1}{2^{2d+1}} \right)^{|A \setminus (Z_{X} \cup \{s,t\})|} \left(\frac{1}{2^{2d+1}}\right)^{|A \cap Z_{X}|} \leq 2^{n+2} \left( \frac{2^{2d+1}-1}{2^{2d+1}} \right)^{|A|}$$
choices for $X \in \vissets(G,s)$.

Moreover, there are at most $\sum_{i=0}^{|Y|} \binom{|A|}{i} \leq n\binom{|A|}{|Y|}$
choices for $Z_{X} \cap A$. Thus
\begin{equation}\label{eq:vis1}
|\vissets(G,s,t)| \leq 2^{n+2} \cdot \left(\frac{2^{2d+1}-1}{2^{2d+1}}\right)^{|A|} \cdot n \binom{|A|}{|Y|}.
\end{equation}

Let us now estimate $\binom{|A|}{|Y|}$ by Lemma \ref{lem:coeff-bound}. Since $D \leq D'$, $|Y| \leq \frac{nd}{\alpha D}$
and by \eqref{eq:Aup} and \eqref{eq:YA}:

\begin{equation}\label{eq:binomAY}
\binom{|A|}{|Y|} \leq \left(\frac{e|A|}{|Y|}\right)^{|Y|} \leq \left(e \frac{n}{2edD'} \cdot \frac{\alpha D}{nd}\right)^{\frac{nd}{\alpha D}} \leq \left(\frac{\alpha}{2d^2}\right)^{\frac{nd}{\alpha D}} < \alpha^\frac{nd}{\alpha D}.
\end{equation}

By the standard inequality $1-x \leq e^{-x}$ we have that
\begin{equation}\label{eq:2d}
(2^{2d+1}-1)/2^{2d+1} = (1-1/2^{2d+1}) \le e^{-1/2^{2d+1}}.
\end{equation}
Using \eqref{eq:Adown}, \eqref{eq:binomAY} and \eqref{eq:2d} we obtain that
$$\binom{|A|}{|Y|} \left(\frac{2^{2d+1}-1}{2^{2d+1}}\right)^{|A|/2} \leq \exp\left(\frac{nd \ln \alpha}{\alpha D} - \frac{n}{20dD' 2^{2d+1}}\right).$$
Plugging in $\alpha = e^{cd}$ and using the fact that $e^{10d} > 40d^2$ for $d \geq 1$ we obtain:
$$\binom{|A|}{|Y|} \left(\frac{2^{2d+1}-1}{2^{2d+1}}\right)^{|A|/2} \leq
\exp\left(\frac{ncd}{e^{(c-10)d}20d \cdot 2dD} - \frac{n}{20dD' 2^{2d+1}} \right).$$
Since $D' = \max(2d,D) \leq 2dD'$ and $e^{4d} > 2^{2d+1}$ as $d \geq 1$, we get
$$\binom{|A|}{|Y|} \left(\frac{2^{2d+1}-1}{2^{2d+1}}\right)^{|A|/2}
\leq \exp\left(\frac{n}{20dD'2^{2d+1}}\left(\frac{c}{e^{(c-14)d}} - 1\right) \right).$$
Finally, for sufficiently large $c$, as $d \geq 1$, we have $c < e^{(c-14)d}$ and
\begin{equation}\label{eq:1}
\binom{|A|}{|Y|} \left(\frac{2^{2d+1}-1}{2^{2d+1}}\right)^{|A|/2} < 1.
\end{equation}
Consequently, plugging \eqref{eq:1} into \eqref{eq:vis1} and using \eqref{eq:Adown} and \eqref{eq:2d} we obtain:
\begin{align*}
|\vissets(G,s,t)| &< n2^{n+2} \left(\frac{2^{2d+1}-1}{2^{2d+1}}\right)^{|A|/2} \\
    & \leq n2^{n+2} \exp\left(-\frac{n}{2^{2d+1} \cdot 20dD'}\right) \\
    & \leq n2^{n+2} \exp\left(-\frac{n}{2^{2d+1}\cdot 20d \cdot e^{e^{cd}}}\right).
\end{align*}
This concludes the proof of the lemma.
Note that the dependency on $d$ in the final constant $\eps_d$ is doubly-exponential.
\end{proof}

\section{Algorithm for TSP}
\label{sec:tsp}

To prove Theorem \ref{thm:tsp}, it suffices to solve in $\Ohstar(2^{(1-\eps_d)n})$ time the following problem.
We are given an undirected $n$-vertex graph $G=(V,E)$ of average degree at most $d$,
vertices $a,b \in V$ and a weight function $c : E \to \mathbb{R_+}$.
We are to find the cheapest Hamiltonian path between $a$ and $b$ in $G$, or verify that no Hamiltonian $ab$-path exists.

We solve the problem by the standard dynamic programing approach.
That is for each $a \in X \subseteq V$ and $v \in X$ we compute 
$t[X][v]$, which is the cost of the cheapest path from $a$ to $v$ with the vertex set $X$.
The entry $t[V][b]$ is the answer to our problem.
Note that it is enough to consider only such pairs $(X,v)$,
     for which there exists an $av$-path with the vertex set $X$.

We first set $t[\{a\}][a] = 0$. Then iteratively,
for each $i=1,2,\ldots,n-1$,
for each $u \in V$,
for each $X \subseteq V$ such that $|X|=i$, $a,u \in X$ and $t[X][u]$ is defined,
for each edge $uv \in E$ where $v \not\in X$,
if $t[X \cup \{v\}][v]$  is undefined or $t[X \cup \{v\}][v] > t[X][u] + c(uv)$,
we set $t[X \cup \{v\}][v] = t[X][u] + c(uv)$.

Finally, note that if $t[X][v]$ is defined then 
$X \setminus \{a,v\} \in \vissets(G,a,v)$.
Hence, the complexity of the above algorithm is within a polynomial factor
from $\sum_{v \in V} |\vissets(G,a,v)|$, which is bounded by $\Ohstar(2^{(1-\eps_d)n})$ by Lemma \ref{lem:vissets-bound}.

\section{Counting Perfect Matchings}

In this section we design algorithms counting
the number of perfect matchings in a given graph.
First, in Section~\ref{sec:in-ex-matchings}
we show an inclusion-exclusion based algorithm,
which given an $n$-vertex graph computes the number
of its perfect matchings in $\Ohstar(2^{n/2})$ time
and polynomial space.
This matches the time and space bounds of the algorithm of 
Bj\"{o}rklund~\cite{bjorklund-matchings}.
Next, in Section~\ref{sec:dp-matchings} we show
how the algorithm from Section~\ref{sec:in-ex-matchings}
can be reformulated as a dynamic programming routine
(using exponential space), which together with Lemma~\ref{lem:vissets-bound}
will imply the running time claimed in Theorem~\ref{thm:matchings}.

\subsection{Inclusion-exclusion based algorithm}
\label{sec:in-ex-matchings}

In the following theorem we show an algorithm
computing the number of perfect matchings of
an undirected graph in $\Ohstar(2^{n/2})$ time and polynomial space,
thus matching the time and space complexity of the
algorithm by Bj\"{o}rklund~\cite{bjorklund-matchings}.

\begin{theorem}
\label{thm:in-ex-matchings}
Given an $n$-vertex graph $G=(V,E)$ in $\Ohstar(2^{n/2})$ time
an polynomial space one can count the number of perfect matchings in $G$.
\end{theorem}

\begin{proof}
Clearly we can assume that $n$ is even.
Consider the edges of $G$ being black and let $V=\{v_0,\ldots,v_{n-1}\}$.
Now we add to the graph a perfect matching of red edges $E_R=\{v_{2i}v_{2i+1} : 0 \le i < n/2\}$
obtaining a multigraph $G'$.
Observe that for any perfect matching $M \subseteq E$ the multiset
$M \cup E_R$ is a cycle cover (potentially with $2$-cycles), where
all the cycles are {\em alternating} - that is when we traverse
each cycle of $M \cup E_R$, the colors alternate (in particular, they have even length).
Moreover, for any cycle cover $Y$ of $G'$ composed
of alternating cycles the set $Y \setminus E_R$
is a perfect matching in $G$.
This leads us to the following observation.

\begin{observation}
The number of perfect matchings in $G$ equals
the number of cycle covers in $G'$ where each cycle is alternating.
\end{observation}

Now we create a directed multigraph graph $G''$ with arcs labeled with elements of $L=\{\ell_0,\ldots,\ell_{n/2-1}\}$,
having $n$ vertices and $2m$ arcs, where $m=|E|$ is the number of black edges of $G'$.
Let $\{v_0'',\ldots,v_{n-1}''\}$ be the set of vertices of the graph $G''$.
For each black edge $v_av_b$ of $G'$ we add to $G''$ two following arcs:
\begin{itemize}
  \item $(v_{a \oplus 1}'',v_b'')$ labeled $\ell_{\lfloor a/2 \rfloor}$,
  \item and $(v_{b \oplus 1}'',v_a'')$ labeled $\ell_{\lfloor b/2 \rfloor}$.
\end{itemize}
By $\oplus$ we denote the XOR operation, that is, 
for any $0 \le x < n$ the vertex $v_{x \oplus 1}$ is the other
endpoint of the red edge of $G'$ incident to $v_x$.

\begin{observation}
The number of cycle covers in $G'$ where each cycle is alternating
equals the number of sets of cycles in $G''$ of total length $n/2$, where
each label $\ell_i$ (for $0 \le i < n/2$) is used exactly once.
\end{observation}

%Now we want to shrink the graph $G'$ into a directed graph $G''$, which has only $n/2$ vertices.
%Let $\{v_1'', \ldots, v_{n/2}''\}$ be the set of vertices of $G''$.
%For each black edge $v_av_b$ of $G'$ we add to the graph $G''$ two arcs $(v_f(a)'',)$ and $()$.

%\begin{observation}
%The number of cycle covers in $G'$, where each cycle is alternating and of even length
%is equal to the number of cycle covers in $G''$.
%\end{observation}

We are going to compute the of sets of cycles in $G''$ where each label is used exactly
once using the inclusion-exclusion principle.

For a vertex $v_a''$ of $G''$, we say that a closed walk $C$ is {\em{$v_a''$-nice}}
if $C$ visits $v_a''$ exactly once and does not visit any vertex $v_b''$ for $b < a$.
A closed walk is {\em{nice}} if it is $v_a''$-nice for some $v_a''$; note that, in this case,
the vertex $v_a''$ is defined uniquely.
For a positive integer $r$ let us define the universe $\Omega_r$ as 
the set of $r$-tuples, where each of the $r$ coordinates
contains a nice closed walk in $G''$ and the total length of all the walks equals $n/2$.
For $0 \le i < n/2$ let $A_{r,i} \subseteq \Omega_r$ be the 
set of $r$-tuples, where at least one walk
contains an arc labeled $\ell_i$.
Note that by the observations we made so far
the number of perfect matchings in $G$ equals
$\sum_{1 \le r \le n/2} |\bigcap_{0 \le i < n/2} A_{r,i}| / r!$,
as the tuples in $\Omega_r$ are ordered and in any tuple of $\bigcap_{0 \le i < n/2} A_{r,i}$ all walks are pairwise different.
Therefore from now on we assume $r$ to be fixed.
By the inclusion-exclusion principle 
$$\left| \bigcap_{0 \le i < n/2} A_{r,i} \right| =  \sum_{I \subseteq \{0,\ldots,n/2-1\}} (-1)^{|I|} \left| \bigcap_{i \in I} (\Omega_r \setminus A_{r,i})\right| \,$$
hence to prove the theorem it is enough to compute the value $|\bigcap_{i \in I} (\Omega_r \setminus A_{r,i})|$
for a given $I \subseteq \{0,\ldots,n/2-1\}$ in polynomial time.
Let $G''_I$ be the graph $G''$ with all the arcs with a label from $L_I=\{\ell_i : i\in I\}$ removed.
Let $p_{a,j}$ be the number of $v_a''$-nice closed walks in $G''_I$ of length $j$.
Note that the value $p_{a,j}$ can be computed in polynomial time by standard dynamic programming algorithm,
filling in a table $t_p[b][i]$, $a \leq b < n/2, 0 \leq i < j$, where $t_p[b][j]$ is the number of walks $W$ from
$v_a''$ to $v_b''$ in $G''$ of length $i$ that visit $v_a''$ only once and does not visit any vertex $v_c''$ for $c < a$.

Finally, having the values $p_{a,j}$ is enough to compute $|\bigcap_{i \in I} (\Omega_r \setminus A_{i})|$
by the standard knapsack type dynamic programming. That is, we fill in a table
$t[q][i]$, $0 \leq q \leq r$, $0 \leq i \leq n/2$, where $t[q][i]$ is the number of $q$-tuples
of nice closed walks in $G''_I$ of total length $i$.
%\todo{Rozwinac DP w appendixie?}

\end{proof}

\subsection{Dynamic programming based algorithm}
\label{sec:dp-matchings}

To prove Theorem~\ref{thm:matchings} we want to reformulate 
the algorithm from Section~\ref{sec:in-ex-matchings},
to use dynamic programming instead of the inclusion exclusion principle.
This causes the space complexity to be exponential, however
it will allow us to use Lemma~\ref{lem:vissets-bound} to
obtain an improved running time for bounded average degree graphs.

Assume that we are given an $n$-vertex undirected graph $G=(V,E)$, where $n$
is even, and we are to count the number of perfect matchings in $G$.
We are going to construct an undirected multigraph $G'$ having only $n/2$
vertices, where the edges of $G'$ will be labeled with unordered pairs
of vertices of $G'$, i.e. with edges of $G$.
As the set of vertices of $G'=(V',E')$ we take $V'=\{v_0',\ldots,v_{n/2-1}'\}$.
For each edge $v_av_b$ of $G$ we add to $G'$ exactly one edge:
$v_{\lfloor a/2 \rfloor}'v_{\lfloor b/2 \rfloor}'$ labeled with $\{v_a, v_b\}$.
For an edge $e' \in E'$ by $\ell(e')$ let us denote the label of $e'$.
Note that $G'$ may contain self-loops and parallel edges.
Observe that if the graph $G$ is of average degree $d$, then
the graph $G'$ is of average degree $2d$.

%Intuitively we can think of the vertices $V'$ as middle
%points of the red edges of $G'$ from the proof
%of Theorem~\ref{thm:in-ex-matcihngs}, where a single edge of $G'$
%corresponds to traversing half of a red edge, followed by a black edge,
%followed by a half of a red edge.
In what follows we count the number of particular cycle
covers of $G'$, where we use the labels of edges
to make sure that a cycle going through a vertex $v_i' \in V'$
never uses two edges of $G'$ corresponding to two edges of $G$ incident to the same vertex.

\begin{lemma}
The number of perfect matchings in $G$
equals the number of cycle covers $C \subseteq E'$ of $G'$,
where $\bigcup_{e \in C} \ell(e) = V$.
\end{lemma}

\begin{proof}
We show a bijection between perfect matchings in $G$
and cycle covers $C$ of $G'$ satisfying the condition $\bigcup_{e \in C} \ell(e) = V$.

Let $M$ be a perfect matching in $G$. As $f(M)$ we define $f(M) = \{v_{\lfloor a/2 \rfloor}'v_{\lfloor b/2 \rfloor}' : v_av_b \in M\}$.
Note that $f(M)$ is a cycle cover and moreover $\bigcup_{e \in f(M)} \ell(e) = V$.
In the reverse direction, for a cycle cover $C \subseteq E'$ of $G'$,
consider a set of edges $h(C)$ defined as $h(C)=\{\ell(e) : e \in C\}$.
Clearly the condition $\bigcup_{e \in C} \ell(e) = V$ implies that $h(C)$ is a perfect matching,
and moreover $h=f^{-1}$.  
\end{proof}

Observe, that if a cycle cover $C \subseteq E'$ of $G'$
does not satisfy $\bigcup_{e \in C} \ell(e) = V$, then
there is a vertex $v_i' \in V'$, such that the two edges
of $C$ incident to $v_i'$ do not have disjoint labels.
Intuitively this means we are able to verify the condition $\bigcup_{e \in C} \ell(e) = V$
locally, which is enough to derive the following dynamic programming routine.

\begin{lemma}
\label{lem:matchings-dp}
Once can compute the number of cycle covers $C$ of $G'$ satisfying 
$\bigcup_{e \in C} \ell(e) = V$ in \linebreak $\Ohstar(\sum_{s,t \in V}|\vissets(G',s,t)|)$ time and space.
\end{lemma}

\begin{proof}
An {\em{ordered $r$-cycle cover}} of a graph $H$ is a tuple of $r$ cycles in $H$, whose union is a cycle cover of $H$.
As each cycle cover of $H$ that contains exactly $r$ cycles can be ordered into exactly $r!$ different ordered $r$-cycle covers,
it is sufficient to count,
for any $1 \leq r \leq n/2$,
the number of ordered $r$-cycle covers $C$ in $G'$ such that each two edges in $C$ have disjoint labels.
In the rest of the proof, we focus on one fixed value of $r$.

For $0 \leq q \leq r$ and $X \subseteq V'$ as $t[q][X]$ let us define the number of ordered $q$-cycle covers in $G'[X]$
where each two edges have disjoint labels; note that $t[r][V']$ is exactly the value we need.
Moreover for $0 \leq q < r$, $X \subseteq V'$, $v_a',v_b' \in X$, $a < b$ and $x\in \{v_{2b},v_{2b+1}\}$ 
as $t_2[q][X][v_a'][v_b'][x]$ we define the number of pairs $(C,P)$ where
\begin{itemize}
  \item $C$ is a ordered $q$-cycle cover of $G'[Y]$ for some $Y \subseteq X \setminus \{v_a',v_b'\}$;
  \item $P$ is a $v_a'v_b'$-path with vertex set $X \setminus Y$ that does not contain any vertex $v_c'$ with $c < a$;
  \item any two edges of $C \cup P$ have disjoint labels;
  \item the label of the edge of $P$ incident to $v_a'$ contains $v_{2a}$;
  \item the label of the edge of $P$ incident to $v_b'$ contains $x$.
\end{itemize}

Note that we have the following border values: $t[0][\emptyset]=1$ and $t[0][X]=0$ for $X \neq \emptyset$. 

Consider an entry $t_2[q][X][v_a'][v_b'][x]$, and let $(C,P)$ be one of the pairs counted in it.
We have two cases: either $P$ is of length $1$ or longer. The number of pairs $(C,P)$
in the first case equals $t[q][X \setminus \{v_a',v_b'\}] \cdot |\{v_a'v_b' \in E': \ell(v_a'v_b') = \{v_{2a},x\}\}|$.
In the second case, let $v_c'v_b'$ be the last edge of $P$; note that $c > a$ by the assumptions on $P$. The label of
$v_c'v_b'$ equals $\{v_{2c},x\}$ or $\{v_{2c+1},x\}$.
Thus, the number of elements $(C,P)$ in the second case equals
$\sum_{v_c' \in X \setminus \{v_a',v_b'\}} \sum_{y \in \{v_{2c},v_{2c+1}\}} t_2[q][X \setminus \{v_b'\}][v_a'][v_c'][y \oplus 1]  \cdot |\{v_c'v_b' \in E': \ell(v_c'v_b') = \{y,x\}\}|$,
where for $y=v_{r}'$ we define $y \oplus 1=v_{r \oplus 1}'$.

%Then, iteratively, for each $q=1,2,\ldots,r$, for each $i=1,2,\ldots,n$, for each $X \subseteq V$ such that $|X|=i$,
%we would like to compute first $t_2[q][X][v_a'][v_b'][x]$ for all valid values of $v_a'$, $v_b'$ and $x$, and then the value $t[q][X]$.

Let us now move to the entry $t[q][X]$ and let $C$ be an ordered $q$-cycle cover in $G'[X]$. Again, there are two cases:
either the last cycle of $C$ (henceforth denoted $W$) is of length $1$ or longer. The number of the elements $C$ of the first type
equals $\sum_{v_a' \in X} t[q-1][X \setminus \{v_a'\}] \cdot |\{v_a'v_a' \in E'\}|$.
In the second case, let $v_a'$ be the lowest-numbered
vertex on $W$ and let $e=v_a'v_b'$ be the edge of $W$ where $v_{2a+1} \in \ell(e)$.
Note that both $v_a'$ and $e$ are defined uniquely; moreover, $a < b$ and no vertex $v_c'$ with $c < a$ belongs to $W$.
Thus the number of elements $C$ of the second type
equals $\sum_{v_a',v_b' \in X, a < b} \sum_{x \in \{v_{2b}, v_{2b+1}\}} \, t_2[q-1][X][v_a'][v_b'][x \oplus 1] \cdot |\{v_a'v_b' \in E': \ell(v_a'v_b') = \{v_{2a+1},x\}\}|$.

So far we have given recursive formulas, that allow computing the entries
of the tables $t$ and $t_2$.
However the values $t[q][X]$, $t_2[q][X][v_a'][v_b'][x]$ for $X \not\in \bigcup_{s,t \in V'} \vissets(G',s,t)$
are equal to zero.
The last step of the proof is to show how to perform the dynamic programming computation
in a time complexity within a polynomial factor from the number of non-zero entries of the table.
We do that in a bottom-up manner, that is iteratively, for each $q=1,2,\ldots,r$, for each $i=1,2,\ldots,n$, 
we want to compute the values of non-zero entries $t[q][X]$ for all sets $X$ of cardinality $i$
and then compute the values of non-zero entries $t_2[q][X][*][*][*]$ for all sets $X$ of cardinality $i$.
Having the non-zero entries for the pairs $(q',i')$ where $q' < q$, $i' \le i$
one can compute the list of non-zero entries $t[q][X]$ for $|X|=i$ by investigating
to which recursive formulas the non-zero entries for $(q',i')$ contribute to.
Analogously having the non-zero entries for the pairs $(q',i')$ where $q' \le q$, $i' < i$
we generate the non-zero entries $t_2[q][X][*][*][*]$ for $|X|=i$, which finishes the proof of the lemma.

\end{proof}
Theorem~\ref{thm:matchings} follows directly from the Lemma~\ref{lem:vissets-bound}
together with Lemma~\ref{lem:matchings-dp}.

\section{Counting Perfect Matchings in Bipartite Graphs}
\label{sec:bip-matchings}

In this section we prove Theorem~\ref{thm:bip-matchings}, i.e. show an algorithm counting
the number of perfect matchings in bipartite graphs 
of average degree $d$ in $\Ohstar(2^{(1-1/(3.55d))n/2})$ time,
improving and simplifying the algorithm of Izumi and Wadayama~\cite{japonce}.

Let $G=(V = A \cup B, E)$ be a bipartite graph, where $|A| = |B| = n/2$, and denote $k=n/2$.
Note that we may assume that each vertex in $G$ is of degree at least $2$,
as an isolated vertex causes no perfect matching to exist, while
a vertex of degree $1$ has to be matched to its only neighbour, hence we can reduce
our instance in that case.
Therefore we assume $d \ge 2$.

Let $B_0 \subseteq B$ be a subset containing $\lfloor k/(\alpha d) \rfloor$ vertices
of smallest degree in $B$, where $\alpha \ge 2$ is a constant to be determined later.
Moreover let $A_0 = N(B_0)$ and observe that $|A_0| \le k/\alpha$,
as vertices of $B_0$ are of average degree at most $d$.
We order vertices of $A$, i.e. denote $A = \{a_1,\ldots,a_{k}\}$, 
so that vertices of $A \setminus A_0$ appear before vertices of $A_0$.
In particular for any $1 \le i \le k(1-1/\alpha)$ we have $N(a_i) \cap B_0 = \emptyset$.

Consider the following standard dynamic programming approach.
For $X \subseteq B$ define $t[X]$ as the number
of perfect matchings in the subgraph of $G$ induced by $\{a_1,\ldots,a_{|X|}\} \cup X$.
Having this definition the number of perfect matchings in $G$ equals $T[B]$.
Observe that the following recursive formula allows to compute the entries of the table $t$,
where we sum over the vertex matched to $a_{|X|}$:
$$t[X] = \sum_{v \in N(a_{|X|}) \cap X} t[X \setminus \{v\}]\,,$$
where $t[\emptyset]$ is defined as $1$.

Let us upper bound the number of sets $X$, for which $t[X]$ is non-zero.
If $|X| \le (1-1/\alpha)k$ and $t[X]>0$, then $X \cap B_0 = \emptyset$,
as otherwise each vertex of $X \cap B_0$ is isolated in $G[\{a_1,\ldots,a_{|X|}\} \cup X]$.
Consequently there are at most $2^{k-\lfloor k /(\alpha d) \rfloor} \le 2^{1+(1-1/(\alpha d))k}$
sets $X$ with $t[X] > 0$ of cardinality at most $(1-1/\alpha)k$.
At the same time there are at most $k \binom{k}{\lceil k/\alpha \rceil}$ sets
of cardinality greater than $(1-1/\alpha)k$.
By using the binary entropy function, we get $\binom{k}{\lceil k/\alpha \rceil} = \Ohstar(2^{H(1/\alpha)k})$,
where $H(p) = -p \log_2 p - (1-p) \log_2(1-p)$.
For $d \ge 2$ and $\alpha = 3.55$ we have $2^{H(1/\alpha)} \le 2^{1-1/(\alpha d)}$.
Consequently if we skip the computation of values $t[X]$ for sets $X$ of
cardinality at most $(1-1/\alpha)k$, such that $X \cap B_0 \neq \emptyset$,
we obtain the claimed running time, which finishes the proof of Theorem~\ref{thm:bip-matchings}.

Note that the constant $\alpha=3.55$ can be improved if we have a stronger lower bound
on $d$. However, in our analysis it is crucial that $\alpha > 2$.

\section{Conclusions and open problems}

We would like to conclude with two open problems that arise from our work.
First, can our ideas be applied to obtain an $\Ohstar(2^{(1-\eps)n})$ time algorithm
for computing the chromatic number of graphs of bounded average degree?
For graphs of bounded maximum degree such an algorithm is due to Bj\"{o}rklund et al. \cite{trimmed-fsc}.

Second, can we make a similar improvements as in our work if only polynomial space is allowed?
To the best of our knowledge, this question remains open even in graphs of bounded maximum degree.

\bibliographystyle{splncs03}
\bibliography{ave-degree}

\begin{thebibliography}{10}
\providecommand{\url}[1]{\texttt{#1}}
\providecommand{\urlprefix}{URL }

\bibitem{bellman}
Bellman, R.: Dynamic programming treatment of the travelling salesman problem.
  J. ACM  9,  61--63 (1962)

\bibitem{hamilton}
Bj{\"o}rklund, A.: Determinant sums for undirected hamiltonicity. In: FOCS. pp.
  173--182. IEEE Computer Society (2010)

\bibitem{bjorklund-matchings}
Bj{\"o}rklund, A.: Counting perfect matchings as fast as ryser. In: Rabani, Y.
  (ed.) SODA. pp. 914--921. SIAM (2012)

\bibitem{trimmed-fsc}
Bj{\"o}rklund, A., Husfeldt, T., Kaski, P., Koivisto, M.: Trimmed {M}oebius
  inversion and graphs of bounded degree. Theory Comput. Syst.  47(3),
  637--654 (2010)

\bibitem{tsp-deg}
Bj{\"o}rklund, A., Husfeldt, T., Kaski, P., Koivisto, M.: The traveling
  salesman problem in bounded degree graphs. ACM Transactions on Algorithms
  8(2), ~18 (2012)

\bibitem{ccc}
Cygan, M., Dell, H., Lokshtanov, D., Marx, D., Nederlof, J., Okamoto, Y.,
  Paturi, R., Saurabh, S., Wahlstr{\"o}m, M.: On problems as hard as
  {CNF}-{SAT}. In: IEEE Conference on Computational Complexity. pp. 74--84.
  IEEE (2012)

\bibitem{held-karp}
Held, M., Karp, R.M.: A dynamic programming approach to sequencing problems. J.
  Soc. Ind. Appl. Math.  10,  196--210 (1962)

\bibitem{seth}
Impagliazzo, R., Paturi, R.: On the complexity of k-sat. J. Comput. Syst. Sci.
  62(2),  367--375 (2001)

\bibitem{japonce}
Izumi, T., Wadayama, T.: A new direction for counting perfect matchings. In:
  FOCS. pp. 591--598. IEEE Computer Society (2012)

\bibitem{lms:tw}
Lokshtanov, D., Marx, D., Saurabh, S.: Known algorithms on graphs on bounded
  treewidth are probably optimal. In: Randall, D. (ed.) SODA. pp. 777--789.
  SIAM (2011)

\bibitem{patrascu-williams}
Patrascu, M., Williams, R.: On the possibility of faster sat algorithms. In:
  SODA. pp. 1065--1075 (2010)

\bibitem{ryser}
Ryser, H.: Combinatorial Mathematics. The {C}arus mathematical monographs,
  Mathematical Association of America (1963)

\bibitem{servedio-wan}
Servedio, R.A., Wan, A.: Computing sparse permanents faster. Inf. Process.
  Lett.  96(3),  89--92 (2005)

\bibitem{woeginger01}
Woeginger, G.J.: Exact algorithms for {NP}-hard problems: A survey. In:
  J{\"u}nger, M., Reinelt, G., Rinaldi, G. (eds.) Combinatorial Optimization.
  Lecture Notes in Computer Science, vol. 2570, pp. 185--208. Springer (2001)

\end{thebibliography}

\end{document}